\documentclass[12pt,journal,onecolumn]{IEEEtran}

\usepackage{amsmath,amssymb,epsf,latexsym,graphics,fullpage,makeidx}
\usepackage{graphicx}

\usepackage{bbm}
\usepackage{amsmath,amssymb,amstext,bm, theorem}
\usepackage{graphicx}
\newtheorem{theorem}{Theorem}[section]
\newtheorem{conjecture}{Conjecture}[section]
\newtheorem{prop}[theorem]{Proposition}
\newtheorem{lemma}[theorem]{Lemma}
\newtheorem{remark}[theorem]{Remark}
\newtheorem{definition}[theorem]{Definition}
\newtheorem{corollary}[theorem]{Corollary}

\newenvironment{proof}{\indent {\bf Proof.}}{\hfill$\bf  \Box$\bigskip}

\newcommand{\post}[2]{
\centering \leavevmode
 \includegraphics[width=#2cm]{#1}
 }

\begin{document}

\title{The Missing Piece Syndrome in Peer-to-Peer Communication\protect\thanksref{T1,T2}}

\title{The Missing Piece Syndrome in Peer-to-Peer Communication\thanks{This work was
supported in part by the National Science Foundation under grant NSF ECS 06-21416
and CCF 10-16959.  }}

 \author{
\authorblockN{Bruce Hajek and Ji Zhu  }\\
\authorblockA{Department of Electrical and Computer Engineering \\
 and the Coordinated Science Laboratory  \\
 University of Illinois at Urbana-Champaign}
}

\maketitle

\begin{abstract}
Typical protocols for peer-to-peer file sharing over the Internet divide
files to be shared into pieces.
New peers strive to obtain a complete collection of pieces from
other peers and from a seed.  In this paper we investigate a problem that can occur
if the seeding rate is not large enough.    The problem is that, even if the
statistics of the system are symmetric  in the pieces, there can be symmetry
breaking, with one piece becoming very rare.   If peers depart after obtaining
a complete collection, they can tend to leave before helping other peers
receive the rare piece.  Assuming that peers arrive with no pieces, there
is a single seed, random peer contacts are made, random useful pieces are
downloaded, and peers depart upon receiving the complete file, the system
is stable if the seeding rate (in pieces per time unit) is greater than the arrival rate, and is unstable if the
seeding rate is less than the arrival rate.  The result persists for any piece
selection policy that selects from among useful pieces, such as rarest first,
and it persists with the use of network coding.
\end{abstract}

\section{Introduction}

Peer-to-peer (P2P) communication in the Internet is provided through the
sharing of widely distributed resources
typically involving end users' computers acting as both clients and servers.  
In an unstructured peer-to-peer network, such as {\em BitTorrent}
 \cite{Cohen03}, a file is divided into many pieces. Seeds, which hold all pieces, 
distribute pieces to peers. New peers continually arrive into the network; they
simultaneously download pieces from a seed or other peers and  upload pieces to other peers. Peers
exit the system after they collect all pieces. 

Determining whether a given P2P network is stable can be difficult.  Roughly speaking, the aggregate
transfer capacity scales up in proportion to the number of peers in the network, but it has to be in the
right places.  Many P2P systems have performed well in practice, and they incorporate a variety of
mechanisms to help achieve stability.   A broad problem, which we address in part, is to provide
a better understanding of which mechanisms are the most effective under various network settings.
These mechanisms include
\begin{itemize}
\item   {\em Rarest first piece selection policies,} such as the one implemented in BitTorrent, whereby
peers determine which pieces are rarest among their neighbors and preferentially download such pieces.
\item {\em Tit-for-tat participation constraints,} such as the one implemented in BitTorrent, whereby peers are choked off
from receiving pieces from other peers  unless they upload pieces to those same peers.   This mechanism provides
an important incentive for peers to participate in uploading pieces, but it may also be beneficial in balancing
the distribution of pieces.
\item  {\em Peers dwelling in the network after completing download,} to provide extra upload capacity. 
\item {\em Network coding
 \cite{AhlswedeCaiLiYeung,GkantsidisRodriguez05},}
whereby data pieces are combined to form coded pieces, giving peers numerous ways to collect enough
information to recover the original data file.
\end{itemize}

This paper determines what parameter values yield stability for a simple model of a P2P file sharing network.
The main model does not include the enhancements mentioned in the previous paragraph, but extensions and discussion
regarding the above mechanisms is given.   The model includes a fixed seed in the network
that uploads with a constant rate.    New peers arrive according to a Poisson process, and have no pieces at the time
of arrival.  {\em Random peer contact} is assumed; each peer contacts a randomly selected target peer periodically. 
{\em Random useful piece selection} is also assumed; each peer chooses which piece to download uniformly at
random from the set of pieces that its selected target has and it itself does not have.  As in the BitTorrent system,
we assume that new peers arrive with no pieces; in effect a peer must first obtain a piece from another peer
or the fixed seed before it can begin uploading to other peers. We also assume that peers depart
as soon as they have completed their collection.


In a P2P network, the last few pieces to be downloaded by a peer are often rare in the network, so it usually takes the
peer a long time to finish downloading. This phenomenon has been referred to as the \textit{delay in  endgame mode} \cite{Cohen03} (or \textit{last piece problem}). We refer to the specific situation that there are many peers
in the network and most of them are missing only one piece which is the same for all peers, as the
{\em missing piece syndrome}. In that situation, peers lucky enough to get the missing piece usually depart immediately
after getting the piece,  so their ability to spread the missing piece is limited.

The main result in this paper is to show, as suggested by the missing piece syndrome, that
the bottleneck for stability is the upload capacity of the seed.    Specifically, if the arrival rate of new peers
is greater than the seed upload rate, the number of peers in the system converges to infinity almost surely; if the
arrival rate of new peers is less than the seed upload rate, the system is positive recurrent and the mean number
of peers in the system in equilibrium is finite.   The next section gives the precise problem
formulation, simulation results illustrating the missing piece syndrome, and the
main proposition.   The proposition is proved in
Sections \ref{sec:instability} and \ref{sec:stability}, with the help of some lemmas given in the appendix.
Section \ref{sec:extensions} provides extensions of the result, including consideration of the enhancement
mechanisms mentioned above.   In particular, it is shown that the region of network stability is not
increased if rarest first piece selection policies, or network coding policies, are applied.   Section \ref{sec:extensions}
also provides a conjecture regarding a refinement of the main proposition for the borderline
case when the arrival rate is equal to the seeding rate; it is suggested that whether the system
is stable then depends on the rate that peers contact each other.


The model in this paper is similar to the flat case of the open system of   Massouli\'{e} and Vojnovi\'{c} 
 \cite{MassoulieVojnovic05,MassoulieVojnovic08}.   The model in \cite{MassoulieVojnovic05,MassoulieVojnovic08}
is slightly different in that, rather than having a fixed seed, it assumes that new peers each arrive with a randomly
selected piece.   A  fluid model, based on the theory of density-dependent jump Markov processes (see \cite{Kurtz81}),
is  derived and studied in  \cite{MassoulieVojnovic05,MassoulieVojnovic08}.  It is shown that  there is a finite
resting point of the fluid ordinary differential equation.   The analysis in this paper is different and complementary.
Rather than appealing to fluid limits, we focus on direct stochastic analysis methods, namely using
coupling to prove transience for some parameter values and the Foster-Lyapunov stability criterion to prove
positive recurrence for complementary parameter values.  Furthermore, our work shows the importance
of considering asymmetric sample paths even for symmetric system dynamics.  Forthcoming work described in
 \cite{ZhuHajek11} provides analysis of P2P networks with  peers having pieces upon arrival, as in
 \cite{MassoulieVojnovic05,MassoulieVojnovic08}, and with peers remaining for some time
in the system after obtaining a complete collection.

Some other works related to stability and the missing piece syndrome are the following.
The instability phenomenon identified in this paper was discovered independently by
Norros et al.  \cite{NorrosReittuEirola09}.     Norros et al.  \cite{NorrosReittuEirola09} proved a version
of our main proposition  for a similar model, for the case of two pieces.  In the model of  \cite{NorrosReittuEirola09} a peer
receives one piece on arrival, with the distribution of the piece number (either one or two)
being determined by sampling uniformly from the group consisting of a fixed seed and
the population of peers already in the system.

Menasch{\'e} et al \cite{Menasche_etal10}
pointed out that  in their simulation studies,  their ``smooth download assumption" and ``swarm sustainability"
break down if the seed upload rate is not sufficiently large.  Leskel\"{a} et al. \cite{LeskelaRobertSimatos10}
investigate stability conditions for a single piece file, or a two piece file when the pieces are obtained
sequentially,  when peers remain in the system for some time after obtaining the piece.
The earliest papers to analytically study unstructured peer-to-peer files systems with arrivals
of new peers are  \cite{QiuSrikant04,YangDeVeciana04}.   These papers provide simple models
in which a two dimensional differential equation is used that does not take into account
the stages of service as peers gain more pieces.

\section{Model formulation and simulations}   \label{sec:formulation}

The model in this paper is a composite of models in \cite{MassoulieVojnovic05,MassoulieVojnovic08,YangDeVeciana06}.
It incorporates Poisson arrivals, fixed seed, random uniform contacts, and random useful piece selection,
as follows.   The parameters of the model are an integer $K\geq 1$ and
strictly positive constants $\lambda, \mu,$ and $U_s.$
\begin{itemize}
\item  There are $K$ pieces and  $\mathcal{F}=\{1, \ldots , K\},$  so that $\mathcal{F}$ indexes all the pieces.
\item The set of proper subsets of  $\mathcal{F}$ is denoted by  $\cal C.$ 
\item A peer with set of pieces $c,$  for some $c\in {\cal C},$ is called a {\em type $c$ peer}.
\item A type $c$ peer becomes a type $c \cup \{i\}$ peer if it downloads piece $i$ for some $i\not\in c.$ 
\item A Markov state is $ \mathbf{x}=(x_c :  c \in {\cal C}  ),$ with $x_c$ denoting the number of type $c$ peers, 
$|\mathbf{x}|$ denoting the number of peers in the system, and $\mathcal{S} = \mathbb{Z}_+^{{\cal C}}$ denoting
the state space of the system.
\item Peers arrive exogenously one at a time with no pieces; the times of arrival form a rate $\lambda$  Poisson process.
\item  Each peer contacts other peers, chosen uniformly at random from among all peers,  for opportunities to
download a piece (i.e. pull) from the other peers,  according to a Poisson process of rate $\mu>0.$   Mathematically,
an equivalent assumption is the following.   Each peer contacts other peers, chosen uniformly at random from among all peers, 
for opportunities to upload a piece (i.e. push) to the other peers,  according to a Poisson process of rate $\mu>0.$ 
\item Downloads are modeled as being instantaneous. This assumption is reasonable in the context of the previous assumption.
\item  Random useful piece selection is used, meaning that when a peer of type $c$  has an opportunity to download a
piece from a peer of type $s,$  the opportunity results in no change of state if $s \subset c.$ 
Otherwise, the type $c$ peer downloads one piece selected at random from $s-c$, with all $|s-c|$ possibilities having equal probability.
\item  There is one fixed seed, which at each time in a sequence of times forming a Poisson process of rate $U_s,$
selects a peer at random and uploads a random useful piece to the selected peer.
\item  Peers leave immediately after obtaining a complete collection.
\end{itemize}

Given a state $\mathbf{x}$, let $T_0(\mathbf{x})$ denote the new state resulting from the arrival of a new peer.
Given $c\in {\cal C},$ $1\leq i \leq K$ such that  $i\notin c,$ and a state  $\mathbf{x}$
 such that $x_c\geq 1,$  let $ T_{c,i}(\mathbf{x} ) $ denote the new state resulting from a type $c$ peer downloading piece $i.$
The positive entries of the generator matrix $Q= (q(\mathbf{x},\mathbf{x'}): \mathbf{x},\mathbf{x'}\in  {\cal S} )$ of the
Markov process are given by:
\begin{eqnarray*}
q(\mathbf{x} , T_0(\mathbf{x}) ) & = &  \lambda   \\
q( \mathbf{x},  T_{c,i}(\mathbf{x}) ) &= &  \frac{x_c}{|\mathbf{x}|}  \left(\frac{U_s}{K-|c|}+ \mu \sum_{s: i\in s} 
\frac{x_s}{|s-c|}    \right)  \\&&~~~~~~~~~~~~~~~~~\mbox{if}~x_c>0~\mbox{and}~i\notin c.
\end{eqnarray*} 

To provide some intuition, we present some simulation results.
Figure \ref{fig.peers_sim} shows simulations of the
system for  $U_s=\mu=1$ and $K=40$ pieces.
\begin{figure}[htb]  
\post{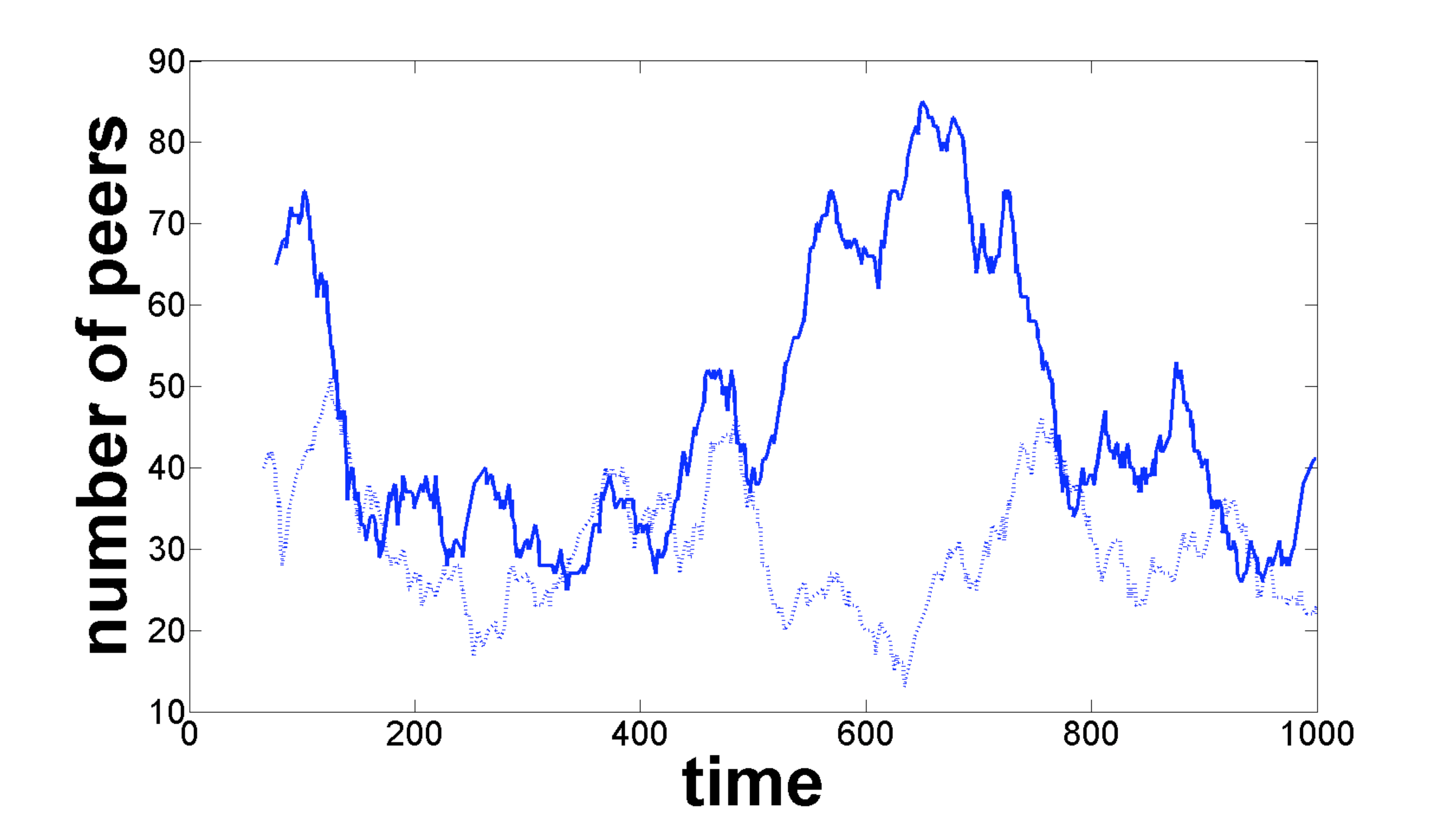}{6}
\post{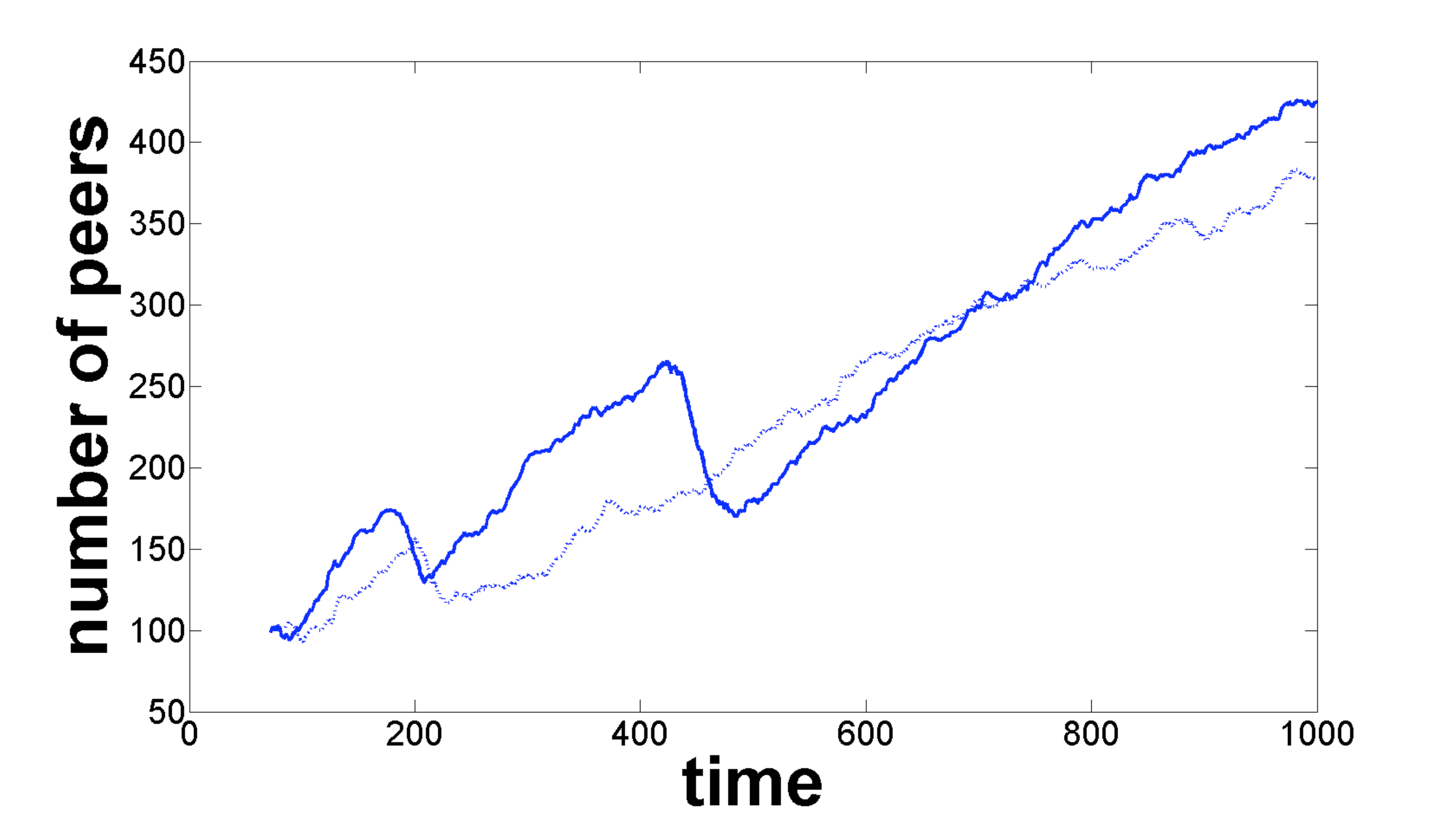}{6}
\caption{Number of peers vs. time.  The first plot is for $\lambda=0.6$ (dashed) and $\lambda=0.8$ (solid),
and the second is for $\lambda=1.2$ (dashed) and $\lambda=1.4$ (solid).  \label{fig.peers_sim}
}
\end{figure}
The first plot shows apparently stable behavior.   After an initial spike, the number
of peers in the system seems to hover around 30 (for $\lambda=0.6$) or 45 (for $\lambda=0.8$), which by Little's law
is consistent with a mean time in system around 50 to 60 time units (or about 25\% to 50\% larger than the
sum of the download times).   However, the second plot shows that for $\lambda=1.2$ or $\lambda=1.4,$
the number of peers in the system does not appear to stabilize, but rather to grow linearly.   The explanation
for this instability is indicated in Figure \ref{fig.pieces_sim}, which shows the time-averaged number of peers that held each
given piece during the simulations, for $\lambda=0.6$ in the first plot and for $\lambda=1.4$ in the second plot.
The first plot shows that the 40 pieces had nearly equal presence in the peers, with piece 7
being the least represented.  The second plot shows that 39 pieces had nearly equal presence and most
of the peers had these pieces most of the time, but only a small number of peers held piece 3.
\begin{figure}[htb]  
 \post{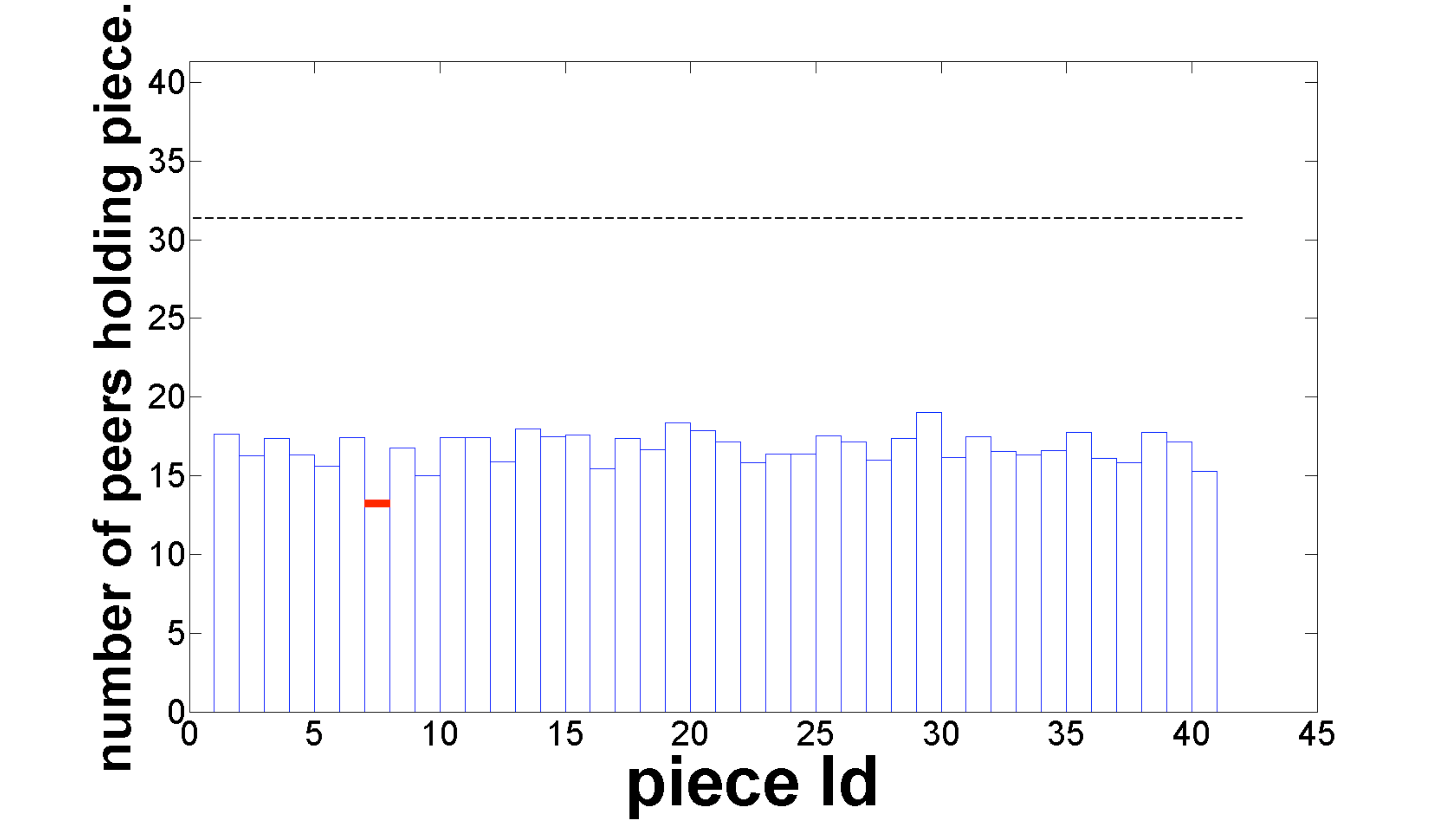}{10}
 \post{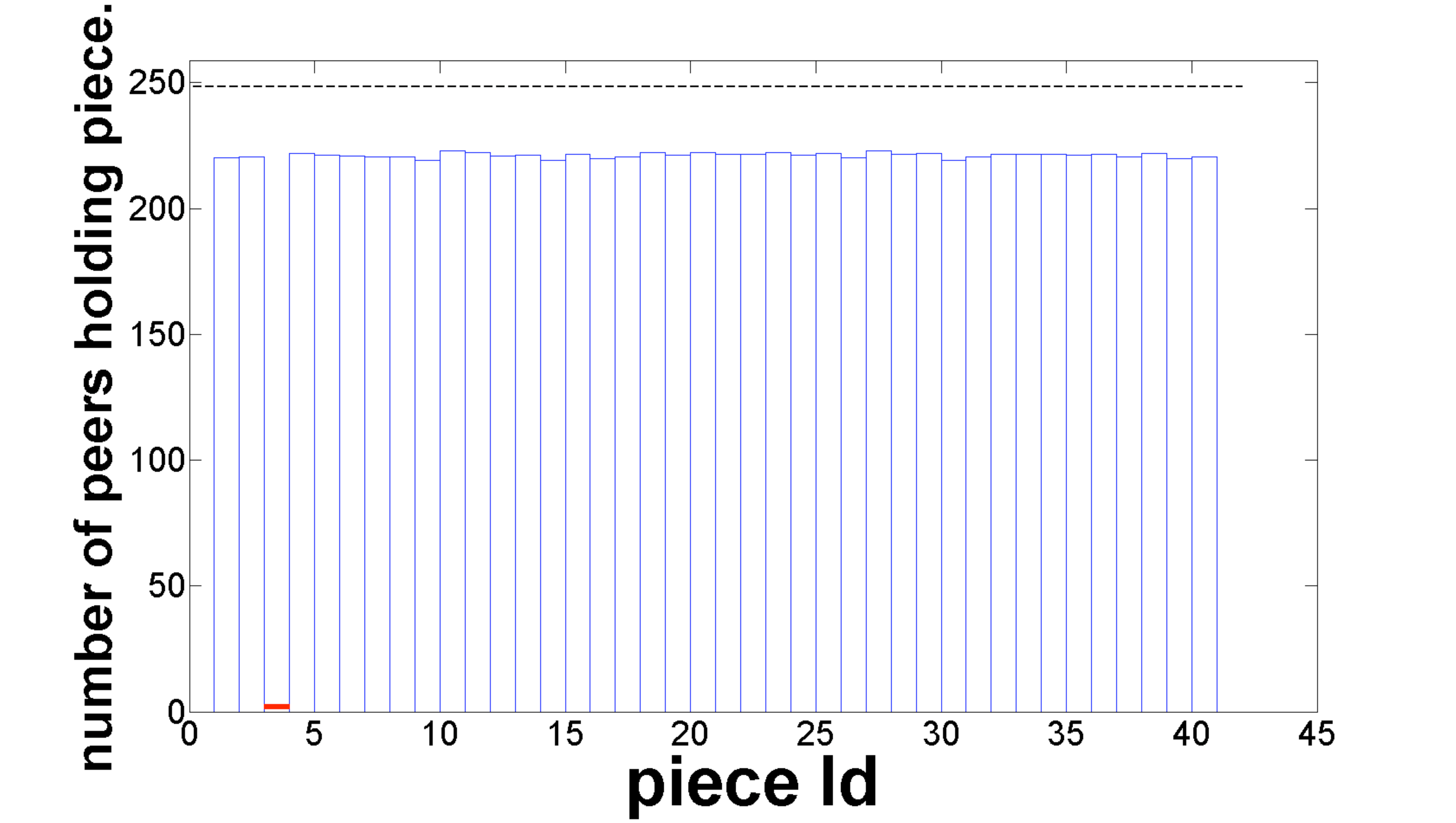}{10}
\caption{Average number of peers holding each piece for the duration of the simulations. The first plot is
for $\lambda=0.6$ and the second is for  $\lambda=1.4.$  The dashed lines indicate time-average
number of peers in system. \label{fig.pieces_sim}
 }
\end{figure}
The following proposition, which is the main result of this paper, confirms that the intuition behind
the simulation results is correct.
\begin{prop}    \label{prop.seed_only}   
 (i)  If  $\lambda > U_s$  then the Markov process
 is transient, and the number of peers in the system converges to infinity with probability one.
 (ii)   If  $\lambda < U_s$  the Markov process with generator $Q$ is positive recurrent, and the equilibrium distribution $\pi$ is
such that $\sum_{ \mathbf{x}}  \pi(\mathbf{x})  |\mathbf{x}| < \infty.$ 
\end{prop}

In the remainder of this section, we give an intuitive explanation for the proposition, which also
guides the proof.
We first give an intuitive justification of Proposition \ref{prop.seed_only}(i), so assume
 $\lambda > U_s.$   Under this condition,  eventually, due to random fluctuations, there  will be
many peers in the system that are all missing the same piece.  While any of the $K$ pieces
could be the missing one, to be definite we focus on the case that the peers are missing
piece one.   A peer is said to be in  the {\em one club}, or to be a one-club peer,  if it has all pieces
except piece one.   We consider the system starting from an initial state in which there are many
peers in the system, and all of them are in the one club. The system then evolves as shown in
Figure \ref{fig.oneclub}.
\begin{figure}[htb] 
\post{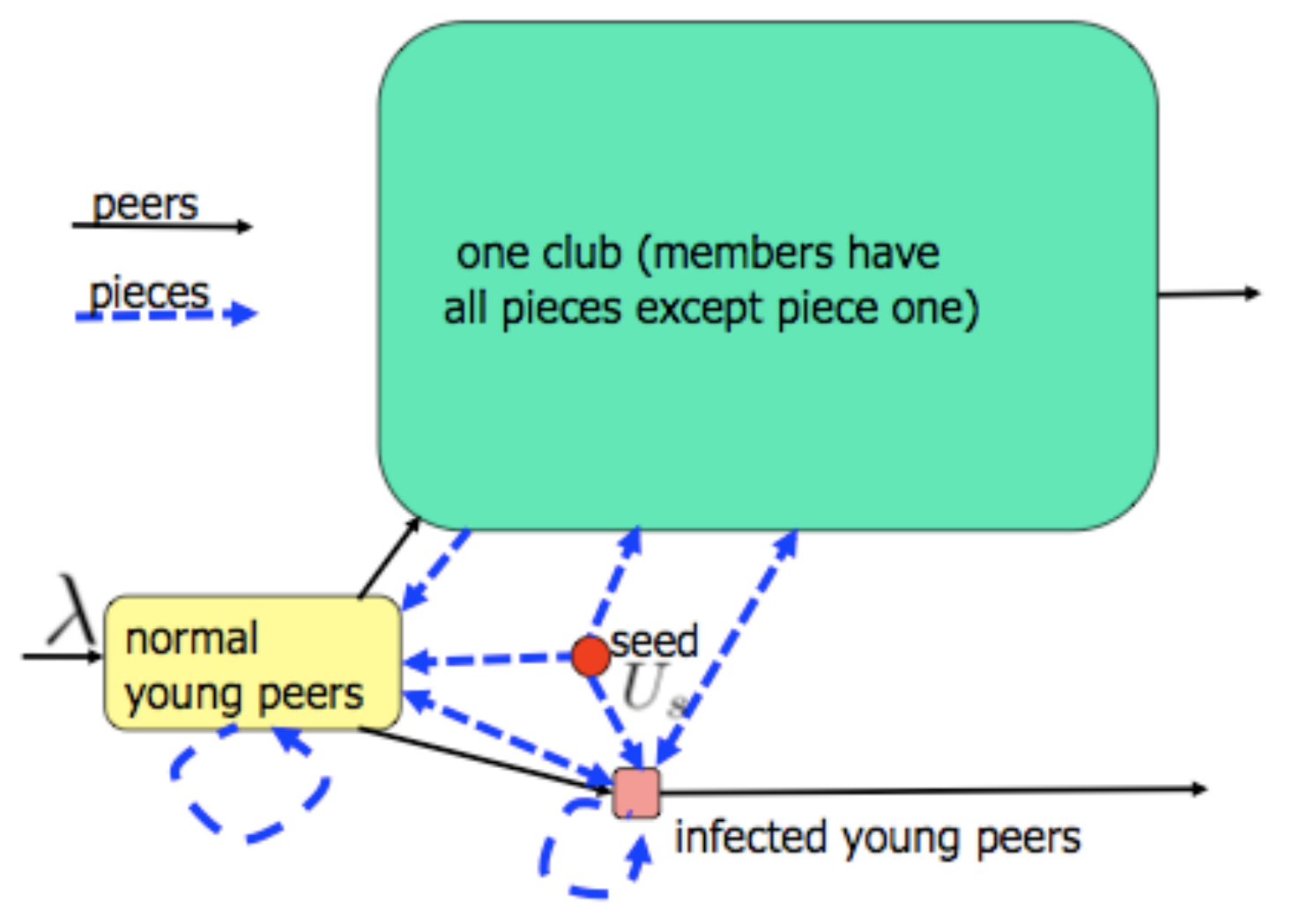}{7}
\caption{Flows of peers and pieces. Solid lines indicate flows of peers; dashed lines indicate flows of pieces.  \label{fig.oneclub}  }
\end{figure}
The large size of the box showing the one club indicates that most peers are one club peers.
A peer not in the one club is said to be a {\em young peer}, and a young peer is said to be
{\em normal} if it does not have piece one and {\em infected} if it does have piece one.
Since there are so many one club peers to download from,  a peer doesn't stay young very long;
most of the young peers join the one club soon after arrival.      However, due to the fixed seed uploading
pieces,  some of the normal young peers become infected peers.   Those infected peers can infect
yet more young peers, thereby forming a branching process.   But typically the infected young peers do not infect
other young peers, so that the branching process is highly subcritical.  Therefore, the rate of departures
from the one club due to uploads of piece one from infected peers is small.  Therefore, most peers
eventually enter the one club, and the main way that peers leave the one club is to receive piece one
directly from the fixed seed.   So the long term arrival rate at the one club is close to $\lambda$ and the
departure rate from the one club is close to $U_s.$   Therefore, the one club can grow at rate
close to  $\lambda  - U_s,$  while the number of young peers will stay about constant.
These ideas are made precise in the proof.

 To understand why the system is stable for  $\lambda < U_s,$
the rough idea is to show that whenever there are many peers in the system, no matter
what the distribution of pieces they hold, the system moves towards emptying out.  If there are
many peers in the system, one of the following two cases holds.  The first case is that most of the peers
have the same number, say $k_o$, of pieces.   Intuitively, the worst case would be for all peers
with $k_o$ pieces to have identical collections of pieces, in which case no peer with $k_o$ pieces
would be useful to another.   However, if $\lambda < U_s,$  such a state can't persist, because peers
with $k_o$ pieces get additional pieces from the fixed seed at an aggregate rate near $U_s,$  while the
long term rate that new peers with exactly $k_o$ pieces can appear is less than or equal to $\lambda.$   
If the system is not in the first case just described,  then there are at least two sizeable groups of peers, so that all the
peers in the first group have one number of pieces and all peers in the second group have some
larger number of pieces.  Then all peers in the second group can be helpful to any peer in the
first group, so that there will be a large rate of downloads.   Thus, if there are many
peers in the system, no distribution of the pieces they hold can persist.   To prove stability,
it is still necessary to show that the state can't spiral out to ever increasing loads through some
quasi-periodic behavior. This is achieved through the use of a potential function and
 the Foster-Lyapunov stability criterion.

\section{Proof of  instability  if  $\lambda > U_s$}  \label{sec:instability}

Proposition \ref{prop.seed_only}(i)  is proved in this section; it can be read independently of the proof
of  Proposition \ref{prop.seed_only}(ii) in the next section.   The proof follows along the lines
of the intuitive explanation given just after the statement of the proposition in Section \ref{sec:formulation},
and an additional explanation of the proof is provided in a remark at the end of the section.
Assume $\lambda > U_s.$ 
If $K=1,$ the system reduces to an $M/M/1$ queueing system with arrival rate $\lambda$ and
departure rate $U_s$, in which case the number of peers in the system converges to infinity with probability one.
So for the remainder of this proof assume $K\geq 2.$   To begin:
\begin{itemize}
\item Select $\epsilon > 0$ so that $3\epsilon < \lambda -U_s.$ 
\item Select $\xi>0$ so that  $\epsilon - 4K \xi U_s > 0,$
and
\begin{equation} \label{eq.rho}
\rho < \frac{1}{2}~~\mbox{where}~~ \rho = 2 \xi  (K-1).
\end{equation}
  It follows from \eqref{eq.rho} that $\xi < 0.5.$
\item Select $\epsilon_o$ small enough that
$
\frac{\epsilon_o}{\lambda-U_s-3\epsilon} < \xi.
$
\item Select $B$ large enough that
\begin{eqnarray} 
\frac{e^{ \lambda[2(K-1)/\mu + 1]}2^{-B}}{1-2^{-\epsilon_o}}\leq 0.1,  \label{eq.BMMinfty}  \\
\frac{64K^2 \xi U_s}{2B(\epsilon-4K\xi U_s )} \leq 0.1,    \label{eq.BMM1}  \\
\frac{\lambda}{2B\epsilon} \leq 0.1,~~\mbox{and}~~\frac{U_s}{2B\epsilon} \leq 0.1. \label{eq.BPoi}
\end{eqnarray}
\item  Select $N_o$ large enough that  $\frac{B}{N_o-3B}\leq \xi$.
\end{itemize}

We shall use the notions of one club, young peer, and infected young peer, as
described in the paragraph after Proposition \ref{prop.seed_only}.
For a given time $t\geq 0,$  define the following random variables:
\begin{itemize}
\item $A_t:$ cumulative number of arrivals, up to time $t$
\item $N_t:$ number of peers at time $t$
\item $Y_t:$ number of young peers at time $t$
\item $D_t:$ cumulative number of uploads of piece one by infected peers, up to time $t$
\item $Z_t:$  cumulative number of uploads of piece one by the fixed seed, up to time $t$
\end{itemize}

The system is modeled by an irreducible, countable-state Markov process. A property of such random processes is
that either all states are transient, or no state is transient.  Therefore, to prove Proposition \ref{prop.seed_only}(i),
it is sufficient to prove that some particular state is transient.    With that in mind, we assume that the initial state is the one with
$N_o$ peers, and all of them are one-club peers.   Let $\tau$ be the extended stopping time defined by
 $\tau=\min\{t \geq 0 : Y_t  \geq \xi N_t  \},$ with the usual convention that $\tau=\infty$ if   $Y_t  < \xi N_t$ for all $t.$
 It suffices to prove that
\begin{equation}  \label{eq:target}
P\{ \tau = \infty ~~\mbox{and} \lim_{t\rightarrow\infty} N_t=+\infty \}  \geq 0.6.
\end{equation}

The equation \eqref{eq:target} depends on the transition rates of the system out of states such that $Y< \xi N .$    Thus, we can and will
prove \eqref{eq:target} instead for an alternative system, that has the same initial state, and the same
out-going transition rates for all states such that  $Y <   \xi   N,$   as the original system.   
The alternative system is defined by modifying the original system by letting the rate of downloads from the set of one-club peers
by each young peer be  $\mu \max\{\frac{N-Y}{N},{1 \over 2}\}$, and the aggregate rate of downloads from the fixed seed to the
set of young peers be $U_s \min\{\frac{Y}{N},  \xi  \}.$     Note that the rates used for this definition
are equal to the original ones on the states such that
$Y< \xi N,$ as required.    The alternative system has the following two properties:
\begin{enumerate}
\item Each young peer receives opportunities to download from one-club peers at rate greater than or equal to  $\mu/2.$
\item The fixed seed contacts the entire population of young peers at aggregate rate less than or equal to $\xi U_s.$
\end{enumerate}
{\em For the remainder of this proof we consider the alternative
system, but for brevity of notation, use the same notation for it as for the original system, and refer to it as the original system.}

The following four inequalities will be established, for $\epsilon, \xi, \epsilon_o, B,$ and $N_o$ satisfying the
conditions given near the beginning of the section.
\begin{eqnarray}
P\{ A_t >   -B+  (\lambda  -  \epsilon) t ~~\mbox{for all}~ t\geq 0 \}    \geq 0.9   \label{eq.Aineq}   \\
P\{ Z_t  <  B  +  (U_s   +  \epsilon) t  ~~ \mbox{for all}~ t\geq 0 \}    \geq 0.9   \label{eq.Zineq}   \\
P\{ Y_t  <   B  +  \epsilon_o t ~~ \mbox{for all}~ t\geq 0 \}    \geq 0.9   \label{eq.Yineq}   \\
P\{ D_t  <  B  +  \epsilon t  ~~\mbox{for all}~ t\geq 0 \}    \geq 0.9  \label{eq.Dineq} 
\end{eqnarray}

Let $\cal E$ be the intersection of the four events on the left sides of  \eqref{eq.Aineq}-\eqref{eq.Dineq}.  Since $N_t$ is greater than
or equal to the number of peers in the system that don't have piece one,
on $\cal E$, \\ $N_t \geq N_o+A_t -D_t-Z_t >  N_o - 3B + (\lambda-U_s-3\epsilon) t $ for all $t\geq 0.$    Therefore, on $\cal E$, for
any $t\geq 0,$
\begin{eqnarray*}
\frac{Y_t}{N_t}  &  <    & \frac{ B  +  \epsilon_o t  }{ N_o  - 3B + (\lambda-U_s-3\epsilon) t }  \\
&\leq & \max \left\{  
\frac{ B   }{ N_o -3B} , ~
 \frac{  \epsilon_o   }{\lambda-U_s-3\epsilon} 
\right\}  \leq \xi.
\end{eqnarray*}
Thus,  $\cal E$ is a subset of the event in \eqref{eq:target}.   Therefore, if  \eqref{eq.Aineq}-\eqref{eq.Dineq} hold,
$P\{{\cal E} \} \geq 0.6,$  and \eqref{eq:target} is implied.   So to complete the proof, it remains to prove
 \eqref{eq.Aineq}-\eqref{eq.Dineq}.
 
The process $A$ is a Poisson process with rate $\lambda$, and $Z$ is stochastically dominated by a Poisson process
with rate $U_s$.   Thus, both \eqref{eq.Aineq} and \eqref{eq.Zineq}  follow from Kingman's moment bound
(see Lemma \ref{lemma:Kingman} in the appendix) and the conditions in \eqref{eq.BPoi} on $B.$

Turning next to the proof of \eqref{eq.Yineq}, we shall use the following observation about {\em stochastic
domination} (the notion of stochastic domination is reviewed in the appendix).  The observation
is a mathematical version of the statement that the number of young peers remains roughly bounded because
peers don't stay young for long.

\begin{lemma}   \label{lemma.MG1compare}
The process $Y$ is stochastically dominated by the number of customers in an $M/GI/\infty$ queueing system
with initial state zero,  arrival rate $\lambda,$ and service times having the Gamma distribution with parameters  $K-1$ and
$\mu/2.$
\end{lemma}

\begin{proof}
The idea of the proof is to show how, with a possible enlargement of the
underlying probability space, an $M/GI/\infty$  system can be constructed on the same probability
space as the original system, so that for any time $t$, $Y_t$ is
less than or equal to the number of peers in the $M/GI/\infty$ system.   Let the $M/GI/\infty$ system have
the same arrival process as the original system--it is a Poisson process of rate $\lambda.$
For any young peer, the intensity of downloads from the one club (i.e. from any peer in the one club) is
always greater than or equal to $\mu/2$ for the original system,  where we use the fact  $1-\xi > 1/2,$
which is true by \eqref{eq.rho} and the assumption $K\geq 2.$   We can thus suppose that
each young peer has an  internal Poisson clock, which generates ticks at rate $\mu/2$, and is such that whenever the internal clock of
a young peer ticks, that young peer downloads a piece from the one club.    We declare that a peer remains
in the $M/GI/\infty$ system until its internal clock ticks $K-1$ times.  This gives the correct service
time distribution, and the service times of different peers in the $M/GI/\infty$ system are independent, as required.
A young peer can possibly leave the original system sooner than it leaves the $M/GI/\infty$ system,
because a young peer in the original system can possibly download pieces at times when its internal clock doesn't tick. 
But if a young peer is still in the original system, it is in the $M/GI/\infty$ system.
\end{proof}

Given this lemma,  \eqref{eq.Yineq}  follows from Lemma  \ref{lemma.mginfty} with $m$ in the lemma equal to
$2(K-1)/\mu$, and $\epsilon$ in the lemma equal to $\epsilon_o,$  and \eqref{eq.BMMinfty}.
It remains to prove \eqref{eq.Dineq}.\\

Consider the following construction of a stochastic system that is similar to the original one, with random variables that
have similar interpretations, but with different joint distributions.   We call it the  {\em comparison system.}
It focuses on the infected peers and the uploads by infected peers, and it is specified in Table \ref{table.compare}. 

\begin{table}[hbt] \caption{Specification of comparison system}  \label{table.compare}
 \begin{tabular}{|p{2.2in}|p{2.3in}|} \hline
~~~~~~~~~~~Original system & ~~~~~~~~~~~Comparison system \\ \hline
The fixed seed creates infected peers at a rate less than $\xi U_s.$
&
The fixed seed creates infected peers at rate  $\xi U_s$.  \\  \hline
An infected peer creates new infected peers at a rate
less than $\xi \mu.$   
&
An infected peer creates new infected peers at rate $\xi \mu.$  \\  \hline
An infected peer uploads piece one to
one-club peers at a rate less than or equal to $\mu.$   
&
An infected peer uploads piece one to
one-club peers at rate $\mu.$  \\   \hline
Just after a peer becomes infected, it requires at most $K-1$
additional pieces, and the rate for acquiring those pieces is
greater than or equal to $ \mu/2.$
&
After a new infected peer arrives, it must download $K-1$
additional pieces, and the rate for acquiring those pieces is  $ \mu/2.$ \\ \hline
\end{tabular}
\end{table}

It should be clear to the reader that both the original system and the comparison system can be
constructed on the same underlying probability space such that any infected peer in the original system
at a given time is also in the
comparison system.    When such a peer becomes infected in the original system,
we require that it also arrives to the comparison system, it discards all pieces it may have
downloaded before becoming infected,  and it subsequently ignores all opportunities to download except
those occurring at the times its internal clock (described in the proof of Lemma \ref{lemma.MG1compare}) ticks.
Because infected young peers possibly stay longer in the comparison system than in the original system, some of
the peers in the comparison system correspond to peers that already departed from the original system.   There can also be
some infected peers in the comparison system that never existed in the original system because the arrival rate of
infected peers to the comparison system is greater than the arrival rate for the original system.
But whenever there is an infected peer in the original system, that peer is also in the comparison system,
and the following property holds.  Whenever any one of the following events happens in the original system, it
also happens in the comparison system:
\begin{itemize}
\item  The fixed seed creates an infected peer.
\item An infected peer creates an infected peer
\item  An infected peer uploads piece one to a one-club peer
\end{itemize}
Events of the second and third type just listed correspond to the two possible ways that infected peers can upload
piece one.  Therefore, the property implies the following lemma, where $\widehat{D}$ is the cumulative number of uploads
of piece one by infected peers, up to time $t$, in the comparison system.

\begin{lemma}  \label{lemma.onecompare}
The process $(D_t: t\geq 0)$ is stochastically dominated by $(\widehat{D}_t: t\geq 0).$
\end{lemma}

We can identify two kinds of infected peers in the comparison system--the {\em root peers}, which are those created by
the fixed seed, and the infected peers created by other infected peers.   We can imagine
that each root peer affixes its unique signature on the copy of piece one that it receives from the fixed seed. 
The signature is inherited by all copies of piece one subsequently generated from that piece through
all generations of the replication process, in which infected peers upload piece one when creating new infected peers. 
In this way, any upload of piece one by an infected peer can be traced back to a unique root
peer.   In summary, the jumps of $\widehat{D}$ can be partitioned according to which root peer
generated them.  Of course, the jumps of  $\widehat{D}$ associated with a root peer happen after the root peer arrives. Let
$(\widehat{\widehat D}_t : t\geq 0)$ denote a new process which results when all of the uploads of piece one generated by a root peer (in the
comparison system)  are counted at the arrival time of the root peer.    Since $\widehat{\widehat D}$ counts the same events as $\widehat{D},$ but does so
earlier,   $\widehat{D}_t \leq   \widehat{\widehat D}_t$  for all $t\geq 0.$   In view of this and Lemma \ref{lemma.onecompare}, it is sufficient to
prove \eqref{eq.Dineq} with $D$ replaced by $\widehat{\widehat D}.$

The random process $\widehat{\widehat D}$ is a compound Poisson process.   Jumps occur at the arrival times of  root peers
in the comparison system, which form a Poisson process of rate $\xi U_s.$ 
Let $J$ denote the size of the jump of $\widehat{\widehat D}$ associated with a typical root peer.  The distribution of $J$ can be
described by referring to an $M/GI/1$ queueing system with arrival rate $\xi \mu$ and service times having the
distribution of a random variable $\widehat {X}$ which has the Gamma distribution with parameters $K-1$ and
$\mu/2.$   Note that $\rho$ in \eqref{eq.rho} is the usual load factor for the reference queueing system:
$\rho = \xi \mu E[\widehat {X}].$ 
The reference queueing system is similar to the number of infected peers in the comparison system, except that the customers in
the $M/GI/1$ queueing system are served one at a time.    We have  $J=J_1+J_2,$  where
\begin{itemize}
\item $J_1$ is the number of infected peers that are descendants of the root peer (not counting the root peer
itself.)   That includes peers directly created by the root peer, peers created by peers created by the root peer,
and so on, for all generations.   $J_1$ has the same distribution as the number of customers in
a busy period of the reference queueing system, not counting the customer that started the busy period.
\item $J_2$ is the number of uploads of piece one to one-club peers by either the root peer  or
any of the descendants of the root peer.   The sum of all the times that the root peer and its descendants are
in the comparison system is the same as the duration, $L,$  of a busy period of the reference queueing system.
While in the comparison system, those peers upload piece one to the one club with intensity $\mu.$
So $E[J_2]=\mu E[L]$ and $E[J_2^2]=\mu^2 E[L]^2 + \mu E[L].$
\end{itemize}
Using this stochastic description, the formulas for the busy period in an $M/GI/1$ queueing system (\eqref{eq.N} and \eqref{eq.L} in the
appendix), and the facts $\rho < 1/2,$  $E[\widehat {X}]=2(K-1)/\mu,$ and  $\mbox{Var}(\widehat {X}) =  (K-1)(2/\mu)^2,$ yields
\begin{eqnarray*}
E[J] & = &  E[J_1]+E[J_2]  =  \frac{1+\mu E[\widehat {X}]}{1-\rho} -1\\
&  \leq   & 2[1+2(K-1)] \leq 4K
\end{eqnarray*}
and
$$
E[J_1^2] \leq E[(J_1+1)^2] = \frac{1+(\xi U_s)^2 \mbox{Var}(\widehat {X})}{(1-\rho)^3} \leq \frac{1+\rho^2}{(1-\rho)^3}
$$
\begin{eqnarray*}
E[J_2^2] & =&  E[ E[J_2^2 | L] ] = \mu E[L] + \mu^2 E[L^2] \\
&  = &  \frac{\mu E[\widehat {X}]}{1-\rho} + \frac{\mu^2 E[\widehat {X}^2] }{(1-\rho)^3}
\end{eqnarray*}
\begin{eqnarray*}
\lefteqn{E[J^2]}\\
 & = & E[(J_1 + J_2)^2] \leq 2\{ E[J_1^2] + E[J_2^2] \}  \\
& \leq &  16\{ 2   + \mu E[\widehat {X}] + \mu^2 E[\widehat {X}^2] \}  \\
& = & 16\left\{ 2   + 2(K-1) + 4(K-1) + 4(K-1)^2 \right\} \\
& =  & 16\left\{  4K^2 - 2K  \right \}    \leq   64K^2 
\end{eqnarray*}
Thus, $\widehat{\widehat D}$ is a compound Poisson process with arrival rate of batches equal to $\xi U_s$ and batch
sizes with first and second moments of the batch sizes bounded by $4K$ and $64K^2$
respectively.   Hence,  \eqref{eq.Dineq} with $D$ replaced by $\widehat{\widehat D}$ follows from Corollary \ref{cor:compoundKingman} and
\eqref{eq.BMM1}.   The proof of Proposition \ref{prop.seed_only}(i) is complete.

\begin{remark}
We briefly explain why the comparison system was introduced in the above proof, to provide a better understanding
of the proof technique.  The intuitive idea behind the definition of the comparison system is that it is based on worst
case assumptions regarding the number of peers that are infected by the fixed seed (i.e. the number of root peers) and
the number of uploads of piece one that can be caused by each root peer.   The advantage is then that the arrivals
of root peers form a Poisson process and the total number of uploads of piece one that can be traced back to different
root peers are independent in the comparison system, so that Kingman's bound for compound
Poisson processes, which is a form of the law of large numbers, can be applied.
\end{remark}

\section{Proof of stability if $\lambda < U_s$}  \label{sec:stability}

Proposition \ref{prop.seed_only}(ii) is proved in this section, using the
version of the Foster-Lyapunov stability criterion given in the appendix, and
the intuition given in the last paragraph of Section \ref{sec:formulation}.

If $V$ is a function on the state space $\cal S,$  then $QV$ is the
corresponding drift function, defined by
$
QV(\mathbf{x}) 
=\sum_{\mathbf{y}:\mathbf{y}\neq \mathbf{x}}  q(\mathbf{x},\mathbf{y}) [V(\mathbf{y})-V(\mathbf{x})].
$
If, as usual, the diagonal entries of $Q$ are defined to make the row sums zero, then the drift function is also
given by matrix-vector multiplication:
$QV(\mathbf{x}) = \sum_{\mathbf{y}}  q(\mathbf{x},\mathbf{y}) V(\mathbf{y}).$

Suppose  $\lambda < U_s.$   Given a state $\mathbf{x}$, let $n_i(\mathbf{x})=\sum_{c\in {\cal C} : |c|=i }  x_c.$    That is, $n_i(\mathbf{x})$
is the number of peers with precisely $i$ pieces.
When the dependence on $\mathbf{x}$ is clear, we write $n_i$ instead of $n_i(\mathbf{x}).$   
We shall use the Foster-Lyapunov criteria with the
following potential function:
$V(\mathbf{x})=\sum_{i=0}^{K-1} b_i \Phi_i(\mathbf{x})$ where $b_0, \ldots , b_{K-1}$ are positive constants and
$\Phi_i(\mathbf{x})=\frac{(n_0+ \cdots + n_i )^2}{2}.$

Let $D_i(\mathbf{x})$ denote the sum,  over all $n_i$ peers with $i$ pieces,  of the download rates of those peers.
Since  any peer with $i+1$ or more pieces always has a useful piece for a peer with $i$ pieces, it follows that
$ D_i(\mathbf{x})  \geq d_i (\mathbf{x}) ,$ where
\begin{equation}  \label{eq.ddef}
d_i(\mathbf{x}) = {n_i \left (U_s+\mu\sum_{j=i+1}^{K-1}{n_j }\right) \over |\mathbf{x}| }.
\end{equation}
We shall write $d_i$ instead of $d_i(\mathbf{x}).$     We have
\begin{eqnarray*}
\lefteqn{Q \Phi_i (\mathbf{x}) } \\
&  \leq  &\frac{\lambda \left[ (n_0+ \cdots + n_i +1)^2 - (n_0+ \cdots + n_i)^2\right]}{2}  +\\
&&~~~~~~~  \frac{d_i  \left[ (n_0+ \cdots + n_i  - 1)^2 - (n_0+ \cdots + n_i)^2\right]}{2}    \\
& =  & (\lambda - d_i)  \left[ n_0+ \cdots + n_i \right]  + \frac{\lambda + d_i}{2}  \\
&\leq  &  \lambda   \left[ n_0+ \cdots + n_i  + \frac{1}{2} \right]  - \left(n_i-\frac{1}{2}\right) d_i 
\end{eqnarray*}
Since $QV=\sum_{i=0}^{K-1} b_i Q \Phi_i$ it follows that
\begin{equation} \label{eq.QVbnd}
QV(\mathbf{x}) \leq \frac{a_o \lambda}{2} +\left( \lambda \sum_{i=0}^{K-1} n_i a_i  \right) -  \sum_{i=0}^{K-1} \left(n_i-\frac{1}{2}\right)b_id_i
\end{equation}
where $a_i=b_i+ \cdots + b_{K-1}$ for $0 \leq i \leq K-1.$   In what follows, assume that the constants
$b_0, \dots , b_n$ are chosen so that $1=b_{K-1} < b_{K-2} < \cdots  < b_1 < b_0$  and
\begin{equation}  \label{eq.a_cond_1}
b_i >  \left(  \frac{\lambda  }{U_s - \lambda} \right) a_{i+1}  ~~\mbox{for}~  0\leq i \leq K-2 .
\end{equation}
Since $a_{i+1}=a_i -  b_i,$ \eqref{eq.a_cond_1} is equivalent to
\begin{equation}  \label{eq.a_cond_2}
U_sb_i -\lambda a_i > 0   ~~\mbox{for}~  0\leq i \leq K-2 .
\end{equation}

The following two lemmas and their proofs correspond to the
two cases described in the intuitive description given in the last paragraph of
Section  \ref{sec:formulation}.
\begin{lemma}  \label{lem:d_LF_K}
There exist positive values $\eta,\epsilon,$ and $L$ so that $QV(\mathbf{x})\leq -\epsilon |\mathbf{x}|$ whenever:   $|\mathbf{x}|\geq L$ and, for some $i$, $n_i \geq (1-\eta) |\mathbf{x}|.$
\end{lemma}

\begin{lemma}\label{lem:d_LF_NK}
Let $\eta$ be as in Lemma \ref{lem:d_LF_K}.  There exist positive values $\epsilon'$ and $L'$  so that $QV(\mathbf{x})\leq -\epsilon' |\mathbf{x}|$ whenever:
$|\mathbf{x}|\geq L'$ and, for all $i$,  $n_i \leq (1-\eta)|\mathbf{x}| .$
\end{lemma}

Lemmas \ref{lem:d_LF_K} and \ref{lem:d_LF_NK} imply that $QV(\mathbf{x})< -\min\{\epsilon',\epsilon\} |\mathbf{x}|  $ whenever
$|\mathbf{x}|>\max\{L,L'\}$, so that $Q$ and $V$ satisfy the conditions of Proposition \ref{cor.FosterCompContinuous}
with $f(\mathbf{x})=\min\{\epsilon',\epsilon\} |\mathbf{x}| $  and $g(\mathbf{x}) = B\mathbbm{1}_{\{|\mathbf{x}|| \leq \max\{L,L'\}  \} }$ where
$B=\max \{  QV(\mathbf{x})   :  |\mathbf{x}| \leq \max\{L,L'\} \}.$
Therefore, to complete the proof of Proposition \ref{prop.seed_only}(ii)
 it remains to prove Lemmas \ref{lem:d_LF_K} and \ref{lem:d_LF_NK}. 

\begin{proof} (Proof of Lemma \ref{lem:d_LF_K}.)
It suffices to prove the lemma for an arbitrary choice of $i.$   So fix $i\in\{0,1,2,...K-1\}, $   and consider a state $\mathbf{x}$
such that $n_i/|\mathbf{x}|> 1-\eta$ (and, in particular, $n_i\geq1$).   Then for any $j\neq i$, $n_j/n_i=(n_j/|\mathbf{x}|)(|\mathbf{x}|/n_i)<\frac{\eta}{1-\eta}.$
Use \eqref{eq.ddef} and \eqref{eq.QVbnd} and an interchange of summation
($\sum_{i=0}^{K-1} \sum_{j=i+1}^{K-1} =  \sum_{j=1}^{K-1} \sum_{i=0}^{j-1}$)
 to get
\begin{eqnarray}
\lefteqn{QV(\mathbf{x} )} \nonumber  \\
&\leq& \frac{a_0\lambda}{2} +  n_i\left(a_i+\displaystyle\sum_{j=0,j\neq i}^{K-1}{{n_j\over n_i}a_j}\right)\lambda - \left(n_i-{1\over2}\right)b_i d_i \nonumber\\
&\leq& \frac{a_0\lambda}{2} + n_ia_i\Big(1 + {Ka_0\over a_i}{\eta \over 1-\eta}\Big)\lambda  \nonumber \\
&&  - \left(n_i-{1\over2}\right)b_i{n_i
\left(U_s+\mu\sum_{j=i+1}^{K-1}{n_j}\right) \over |\mathbf{x}| }\nonumber\\
&\leq& \frac{a_0\lambda}{2} \nonumber \\
&&+ n_i\left\{ a_i\Big(1+{Ka_0\over a_i}{\eta \over 1-\eta }\Big)\lambda - b_i(1-\eta)U_s + {b_i U_s \over 2|\mathbf{x}|}\right\} 
\label{eq.QV_bnd}
\end{eqnarray}
Notice that according to \eqref{eq.a_cond_2},
\begin{eqnarray*}
\lim_{\eta \rightarrow  0}\left\{a_i\left(1+{Ka_0\over a_i}
  {\eta \over 1-\eta }\right)\lambda - b_i (1-\eta) U_s\right\}   \\
= a_i\lambda-b_iU_s<0,
\end{eqnarray*}
and
$$\displaystyle\lim_{|\mathbf{x}|\rightarrow\infty}{b_i U_s \over 2|\mathbf{x}|} = 0.$$
Thus, if $\eta$ is small enough and $ | \mathbf{x}|  $ is large enough, the
quantity within braces in \eqref{eq.QV_bnd} is negative.   Therefore,
if $\eta$ and $\epsilon$ are small enough, and $L$ is large enough,
$$
QV(\mathbf{x})\leq\frac{    a_0 \lambda+ n_i  \{ a_i\lambda-b_iU_s \}   }{2}      \leq -\epsilon |\mathbf{x}|
$$
under the conditions of the lemma, whenever $|\mathbf{x}|\geq L.$   Lemma \ref{lem:d_LF_K} is proved.
\end{proof}

\begin{proof} (Proof of Lemma \ref{lem:d_LF_NK}.)
Let $\eta$ be given by Lemma \ref{lem:d_LF_K}, and consider a state $\mathbf{x}$ such that $n_i/|\mathbf{x}|\leq 1-\eta$ for all $i$.
It follows that there exists $i_1$ and $i_2$ with $0\leq i_1 < i_2 \leq K-1$ such that $n_{i_1} \geq \frac{\eta  |\mathbf{x}|}{K}$
and $n_{i_2} \geq \frac{\eta  |\mathbf{x}|}{K}.$  Then
\begin{eqnarray}
\lefteqn{QV(\mathbf{x})} \nonumber   \\
&\leq&\frac{a_0\lambda}{2}+ |\mathbf{x}|a_0 K\lambda-\left(n_{i_1}-{1\over2}\right)b_{i_1}  d_{i_1}\nonumber\\
&=&\frac{a_0\lambda}{2}+|\mathbf{x}| a_0 K\lambda  \nonumber   \\
&& -\left(n_{i_1} -{1\over2}\right)b_{i_1}
{n_{i_1}(U_s+\mu\sum_{j=i_1+1}^{K-1}{n_j}) \over |\mathbf{x}| }\nonumber\\
&\leq&\frac{a_0\lambda}{2}+ |\mathbf{x}| a_0 K\lambda-\left(\frac{\eta|\mathbf{x}|}{K}-{1\over2}\right)b_{i_1}{\eta^2|\mathbf{x}| \over  K^2 }\mu\nonumber\\
&\leq&\frac{a_0\lambda}{2}+|\mathbf{x}|\left\{ a_0 K\lambda +{b_0\mu \over2} \right\}
   -  \left({\eta \over K}\right)^3|\mathbf{x}|^2 \mu  \label{eqn:upper_bound_dV_ni_nj}
\end{eqnarray}
The conclusion of the lemma follows because of the term in \eqref{eqn:upper_bound_dV_ni_nj}  that is quadratic in $|\mathbf{x}|.$
\end{proof}

\section{Generalization and Discussion}   \label{sec:extensions}
\subsection{General Piece Selection Policies}
A piece selection policy is used by a peer to choose which piece to download whenever it contacts another peer.
The random useful piece selection policy is assumed above, but the results extend to a large class of piece
selection policies.   Essentially the only restriction needed is that if the contacted peer has a useful piece for the
contacting peer, then a useful piece must be downloaded.     This restriction is similar to a work
conserving restriction in the theory of service systems.     In particular, the results hold for a broad class of
rarest first piece selection policies.     Peers can estimate which pieces are more rare in a distributed way,
by exchanging information with the peers they contact.   Even more general policies would allow the piece
selection to depend in an arbitrary way on the piece collections of all peers.     Interestingly enough,
the results extend even to seemingly bad piece selection policies.   For example, it includes the
sequential piece selection policy, in which peers obtain the pieces in order, beginning with piece one.
The sequential policy can be viewed as a  {\em most abundant first} useful piece selection policy, or
just the opposite of rarest piece first.

To be specific, consider the following family $\cal H$ of piece selection policies.   Each policy
 in $\cal H$
corresponds to a mapping $h$ from ${\cal C} \times ({\cal C}  \cup \{ {\cal F} \}  ) \times {\cal S}$
to the set of probability distributions on $\cal F,$   satisfying the usefulness constraint:
$$ \sum_{i\in B-A}  h_i(A,B,\mathbf{x})=1~~\mbox{whenever}~ B \not\subset A $$
with the following meaning of $h$:
\begin{itemize}
\item When a type $A$ peer selects a piece to download from a type $B$ peer and the state of the entire network is $\mathbf{x},$
piece $i$ is selected with probability $h_i(A,B,\mathbf{x}),$  for $i \in {\cal F}.$
\item When the fixed seed selects a piece to upload to a type $A$ peer and the state of the entire network is $\mathbf{x},$
piece $i$ is selected with probability  $h_i(A,{\cal F},\mathbf{x}),$  for $i \in {\cal F}.$
\end{itemize}
The piece selection policies noted above are included in $\cal H.$

Reconsider the proof of transience in Section \ref{sec:instability} under a piece selection policy in $\cal H$. 
From any state it is possible to reach the empty state, and from the empty state it is possible to reach a state
with one peer in the network having all pieces except some piece $i_0.$  
From that state, for any $N_o\geq 1$, it is
possible to reach the state with $N_o$ peers missing only piece $i_0,$  and no other peers in the network.
It may be impossible for $i_0$ to equal one, but by
renumbering the pieces if necessary, it can be assumed without loss of generality that $i_0$ is one.  
Thus, whatever piece selection policy in $\cal H$ is applied,  beginning from any initial state, for
any $N_o \geq 1,$  in a finite time with a
positive probability, the system can arrive into the state where there are $N_o$ peers and all of them are one-club peers.
Thus, as in Section \ref{sec:instability}, to prove transience it suffices to show that from such an initial state, there
is a positive probability that the number of peers converges to infinity.
The arrival rate of new peers and the upload rate of the seed does not depend on the piece selection policy, so \eqref{eq.Aineq} and \eqref{eq.Zineq} are valid for any piece selection policies in $\cal H$. Moreover, Lemma \ref{lemma.MG1compare} and Lemma \ref{lemma.onecompare} are valid for any piece selection policies in $\cal H$ because the two lemmas depend on the properties that peer selection is uniformly random and the piece selection is useful if a useful piece is available.   Therefore \eqref{eq.Yineq} and \eqref{eq.Dineq} are also valid for any piece selection policy in $\cal H$. Thus, we conclude that the
proof of Proposition \ref{prop.seed_only}(i) in Section \ref{sec:instability} works for any piece selection policy in $\cal H.$

Reconsider next the proof of positive recurrence in Section \ref{sec:stability}, but for an
arbitrary piece selection policy in $\cal H.$   The inequalities developed for the proofs of
Lemmas \ref{lem:d_LF_K}  and \ref{lem:d_LF_NK} hold with the same Lyapunov function;
useful piece selection suffices.
Thus, if $\lambda < U_s,$  it can be shown that the Lyapunov stability condition,
namely $QV( \mathbf{x} ) \leq -\epsilon | \mathbf{x} |,$  for $|\mathbf{x}|$ sufficiently large,
still holds.    The final conclusion has to be modified, however, because under some policies
in $\cal H,$  the Markov process might no longer be irreducible.   For example, with the sequential
 useful piece selection policy, the set of states such that every peer holds
a set of pieces of the form $\{1,2, \ldots , J\}$ for some $J$ with $0\leq J \leq K-1,$
is a closed subset of states, in the terminology of classification of states of discrete-state Markov
processes.     In general, the set of all states that are reachable from the empty
state is the unique minimal closed set of states, and the process restricted to that set of states
is irreducible.   By a minor variation of the Foster-Lyapunov stability proposition, the Lyapunov
stability condition implies that the Markov process restricted to that closed set of states is
positive recurrent, and the mean time to reach the empty state beginning from an arbitrary initial state is finite.

We summarize the discussion of the previous two paragraphs as a proposition.
\begin{prop}   (Stability conditions for general useful piece selection policies)
Suppose a useful piece selection policy from $\cal H$ is used, for a network
with random peer contacts and parameters $K$, $\lambda$, $U_s,$ and $\mu$ as in Section \ref{sec:formulation}.
There is a single class of closed states containing the empty state, and all other states are transient.  
 (i)  If  $\lambda > U_s$  then the Markov process
 is transient, and the number of peers in the system converges to infinity with probability one.
(ii)     If  $\lambda < U_s$  the Markov process with generator $Q$ restricted to the closed set of states is positive recurrent,
the mean time to reach the empty state from any initial state has finite mean, and the equilibrium distribution $\pi$ is
such that $\sum_{ \mathbf{x}}  \pi(\mathbf{x})  |\mathbf{x}| < \infty.$   
\end{prop}
Thus, with the exception of the borderline case $\lambda=\mu,$  rarest first piece selection does not increase the
region of stability, nor does most abundant first piece selection decrease the region of stability.

\subsection{Network Coding}

Network coding, introduced by Ahlswede, Cai, and Yeung,  \cite{AhlswedeCaiLiYeung},
can be naturally incorporated into P2P distribution networks, as noted in
\cite{GkantsidisRodriguez05}.    The related work \cite{DebMedardChoute06} considers
all to all exchange of pieces among a fixed population of peers through random contacts
and network coding.    The method
can be described as follows.    The file to be transmitted is divided into
$K$ data pieces, $m_1, m_2, \ldots  , m_K.$     The data pieces are
taken to be vectors of some fixed length $r$ over a finite field $\mathbb{F}_q$ with
$q$ elements, where $q$ is some power of a prime number.
If the piece size is $M$ bits, this can be done by viewing each message as an
$r=\lceil M / \log_2(q) \rceil$ dimensional vector over $\mathbb{F}_q.$
Any coded piece $e$ is a linear combination of
the original $K$ data pieces:
$e=\sum_{i=1}^K \theta_i m_i;$
the vector of coefficients $(\theta_1, \ldots , \theta_K)$ is called
the {\em coding vector} of the coded piece; the coding vector
is included whenever a coded piece is sent.
The fixed seed uploads coded pieces to peers, and peers
exchange coded pieces.      In this context, the type of a peer $A$ is the
subspace $V_A$ of $\mathbb{F}_q^K$
spanned by the coding vectors of the coded pieces it has
received.   Once the dimension of $V_A$ reaches $K$, peer $A$
can recover the original message.

When peer A contacts peer B, suppose peer B sends peer
A a random linear combination of its coded pieces, where
the coefficients are independent and uniformly distributed
over  $\mathbb{F}_q.$   Equivalently, the coding vector of the
coded piece sent from $B$ is uniformly distributed over
$V_B.$      The coded piece is considered useful to $A$ if adding
it to $A$'s collection of coded pieces increases the dimension of $V_A.$
Equivalently, the piece from $B$ is useful to $A$ if its coding vector is
not in the subspace $V_A \cap V_B.$    The probability the piece is useful to
$A$ is therefore given by
$$
P\{\mbox{piece is useful}\} = \frac{|V_A\cap V_B|}{|V_B|}
= 1-q^{  dim(V_A\cap V_B)  - dim(V_B) } 
.
$$
If peer $B$ can possibly help peer $A$, meaning $V_B\not\subset V_A$
(true, for example, if $dim(V_B) > dim(V_A)$),
the probability that a random coded piece from B is helpful to A is greater than or equal to $1-\frac{1}{q}.$
The probability a random coded piece from the seed is
useful to a peer $A$ with $dim(V_A)=K-1$ is precisely $1-\frac{1}{q}.$
Therefore, when all peers have the same state and the common state has dimension $K-1$, the departure
rate from the network is $\widetilde{U}_s = U_s(1-\frac{1}{q}).$

The network state $\mathbf{x}$ specifies the number of peers in the network of each type. 
There are only finitely many types, so the overall state space is still countably infinite.  
Moreover, the Markov process is easily seen to be irreducible.

Reconsider the proof of transience  in Section \ref{sec:instability}, but now under network coding.
Fix any subspace $V^-$ of $\mathbb{F}_q^K$ with dimension $K-1.$ 
Call a peer a one-club peer if its
state is $V^-.$    For any $N_o\geq 1$, it is possible to reach the state with $N_o$ one-club peers and no other
peers in the network.
As before, call a peer a young peer if it is not a one-club peer.  In the case of network coding,
call a peer infected if its state is not a subspace of $V^-.$  The only way a peer can become
infected is by downloading a piece either from the seed or from an infected peer.
Lemmas  \ref{eq.Aineq} and \ref{eq.Zineq} are valid for network coding, if the condition $\lambda > U_s$
is replaced by $\lambda >  \widetilde{U}_s.$
Moreover, Lemma \ref{lemma.MG1compare} and Lemma \ref{lemma.onecompare} are valid for network
 coding  because the two lemmas depend on the properties that peer selection is uniformly random and the rate useful
 pieces are delivered by the seed to one-club peers is arbitrarily close to $\widetilde{U}_s.$
 Thus, we conclude that Proposition \ref{prop.seed_only}(i) in Section \ref{sec:instability}, with $U_s$
 replaced by $\widetilde{U}_s,$  extends to the case of network coding.
 
 Reconsider the proof of positive recurrence in Section \ref{sec:stability}, but with random
useful piece selection replaced by network coding as described, and $U_s$ replaced
by $\widetilde{U}_s=U_s(1-\frac{1}{q}).$  Suppose the same Lyapunov function is
used, except the new meaning of $n_i(\mathbf{x}),$ or $n_i$ for short,  is the number of peers
$A$ with $dim(V_A)=i.$
Lemmas  \ref{lem:d_LF_K} and \ref{lem:d_LF_NK} are valid for network coding, if the condition $\lambda < U_s$
is replaced by $\lambda < \widetilde{U}_s.$
Thus, if $\lambda < \widetilde{U}_s,$  it can be shown that the Lyapunov stability condition,
namely $QV( \mathbf{x} ) \leq -\epsilon | \mathbf{x} |,$  for $|\mathbf{x}|$ sufficiently large,
still holds, and the Foster-Lyapunov stability criterion applies.

We summarize the discussion of the previous  two paragraphs as a proposition.
\begin{prop}    (Stability conditions for network coding based system)
Suppose random linear network coding with vectors over $\mathbb{F}_q^K$ is used,
with random peer contacts and parameters $K$, $\lambda$, $U_s,$ and $\mu$ as in Section \ref{sec:formulation}.  
  (i)  If  $\lambda >  U_s(1-\frac{1}{q})$  then the Markov process  is transient, and the number of peers in the system converges to infinity with probability one.
  (ii)     If  $\lambda < U_s(1-\frac{1}{q})$  the Markov process is positive recurrent, and the equilibrium distribution $\pi$ is
such that $\sum_{ \mathbf{x}}  \pi(\mathbf{x})  |\mathbf{x}| < \infty.$  
\end{prop}
Thus, as $q\rightarrow \infty,$   the stability region for the system with network coding converges to that
for useful piece selection.    Network coding has the advantage that no exchange of state information among
peers is needed because there is no need to identify useful pieces.


\subsection{Peer Seeds}
In many unstructured peer-to-peer systems, such as BitTorrent, peers often dwell in the network awhile after they
have collected all the pieces.  In effect, these peers temporarily become seeds, called peer seeds.   The uploading
provided by peer seeds is able to mitigate the missing piece syndrome and enlarge the stability region. Intuitively, if
every peer can upload, on average, just one more piece after collecting all pieces, then every peer
can help one one-club peer to depart, so the missing piece syndrome would not persist. This is
explored for the case of $K=1$ and $K=2$ (for the sequential piece selection policy) in
\cite{LeskelaRobertSimatos10} and  for random useful piece selection with arbitrary  $K\geq 1$  in \cite{ZhuHajek11}.

\subsection{Peer Selection and Tit-for-Tat}
Another way to overcome the missing piece syndrome relies on peer selection policies. For instance, if young peers
contact infected peers preferentially, or if the seed uploads to young peers preferentially, the network can be stabilized
by the resulting increase in the number of infected peers.    So some sort of coordination policy, providing the identification
of rare pieces and young peers, and the transmission of the rare pieces to the young peers, can counter the missing piece syndrome.
A mechanism built into BitTorrent, called tit-for-tat operation,  may alter the peer selection policy enough to yield
stability for any choice of $\lambda, \mu,$ and $U_s.$   Under  tit-for-tat operation, peers upload almost exclusively
to peers from which they can simultaneously download.   An obvious benefit of tit-for-tat is to give peers incentive
to upload, thereby helping other peers, but it also may be effective against the missing piece syndrome.    Specifically,
tit-for-tat encourages one-club peers to reduce their  rate of download to the young peers, because the young peers
have nothing to upload to the one-club members.   This increases the amount of time that peers remain young, giving
them a greater chance to obtain a rare piece  from the fixed seed.   Also, infected peers would preferentially send to
young peers, because often a normal young peer and an infected young peer would be able to help each other.
While it is thus clear that tit-for-tat operation helps combat the missing piece syndrome,  we leave open the problem
of quantifying the effect for a specific model.

\subsection{The Borderline of Stability}
We have shown that, for any $\mu > 0,$  the system is
stable if $\lambda < U_s$ and unstable if $\lambda > U_s,$  
and this result is insensitive to the value of $\mu$ and to the
 piece selection policy, as long as a useful piece is selected whenvever possible.
 While it may not be interesting from a practical point of view,  we comment
on the case  $\lambda = U_s.$  First, we give a precise result
for a limiting case of the original system, and then we offer
a conjecture.

A simpler network model results by taking a limit as $\mu \rightarrow \infty.$  
Call a state {\em slow} if all peers in the system have the same type, which includes
the state such that there are no peers in the system.  Otherwise, call a state {\em fast}. 
The total rate of transition out of any slow state does not depend on $\mu,$ and the
total rate out of any fast state is greater than or equal to $\mu/2.$    For very large values of
$\mu,$  the process spends most of its time in slow states.    The original
Markov process can be transformed into a new one by {\em watching} the
original process while it is in the set of slow states.   This means removing
the portions of each sample path during which the process is in fast states,
and time-shifting the remaining parts of the sample path to leave no
gaps in time.  The limiting Markov process, which we call the
$\mu=\infty$ process, is the weak limit (defined as usual for probability measures
on the space of c\`{a}dl\`{a}g sample paths equipped
with the Skorohod topology)  of the original process
watched in the set of slow states, as $\mu \rightarrow \infty.$
By symmetry of the model, the state space of the $\mu=\infty$ process
can be reduced further, to
$\widehat{\cal S}=\{(0,0)\} \cup \{(n,k): n\geq 1, 1\leq k \leq K-1\},$
where a state $(n,k)$ corresponds to $n$ peers in the system which
all possess the same set of $k$ pieces.  The positive transition rates
of the $\mu=\infty$ process are given by:
$$
\begin{array}{ccc}
transition  &  rate &   condition \\ 
(n,k) \rightarrow (n+1,k)  &  \lambda   & (n,k) \in \widehat{\cal S}  \\
(n,k)  \rightarrow (n,k+1) & U_s & n\geq 1, 0\leq k \leq K-2  \\
(n,K-1)  \rightarrow  (n-1,K-1) & U_s &  n\geq 2, k=K-1  \\
(1,K-1) \rightarrow ~~(0,0)  ~~~~~~~~~&  U_s 
\end{array}
$$
and the transition rate diagram is pictured in Figure \ref{fig.mu_infinity}
for $K=3.$   
\begin{figure}
\post{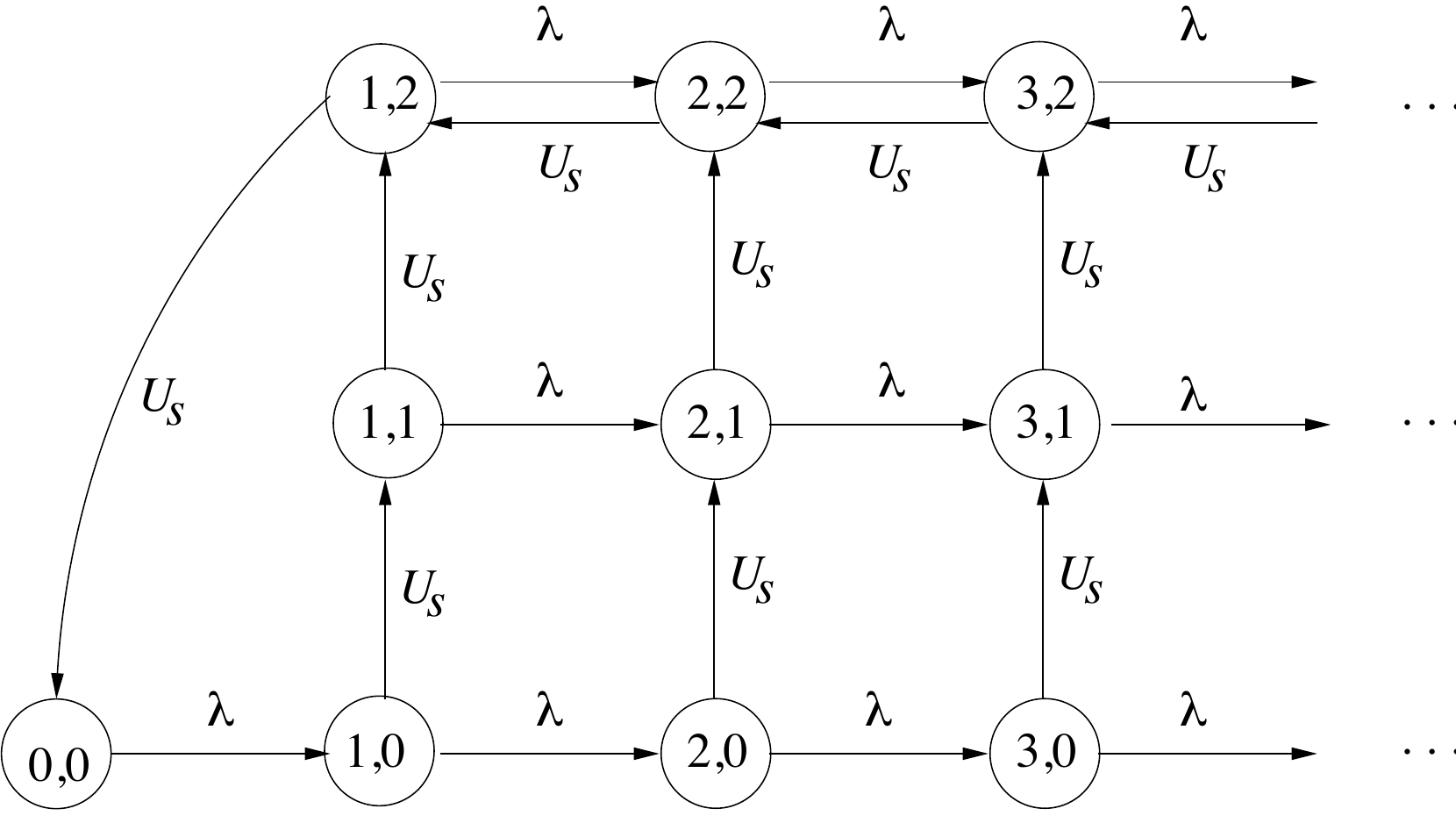}{6}
\caption{Transition rates for $\mu=\infty$ system for $K=3.$}
\label{fig.mu_infinity}
\end{figure}
 The top layer of states consists of those for which all peers have
$K-1$ pieces.   These states correspond to all peers being in the one club,
or all missing some other piece.  From any state the process reaches the top
layer in mean time less than or equal to $\frac{1}{\lambda} + \frac{K-1}{U_s},$  and within the
top layer  the process behaves like a birth-death process with
birth rate $\lambda$ and death rate $U_s.$   Since such  birth-death processes
are null-recurrent if $\lambda=U_s,$  it follows that {\em the $\mu=\infty$ model
is null-recurrent if $\lambda = U_s.$}

Consider the original process with $\lambda=U_s$ and finite $\mu.$   Suppose
the process is in a state with a very large one club which includes all or nearly
all the peers; let $n$ be the number of peers in the one club. 
New young peers arrive at rate $\lambda$, and they are in the system
for approximately $\frac{1}{\mu}$ time units while they are holding exactly $k$ pieces
for $0 \leq k \leq K-2.$ .  Thus, over the short term,  the mean number of young peers in
the system holding $k$ pieces is near $\frac{\lambda}{\mu}$ for $0 \leq k \leq K-2.$  
The average fraction of peers that are young peers holding $k$ pieces is thus approximately
$\frac{\lambda}{n\mu}$  for $0 \leq k \leq K-2.$  
The average total rate that young peers holding $k$ pieces become infected is dominated by the
rate the fixed seed downloads piece one to them and is thus approximately $\frac{U_s\lambda}{(K-k) n\mu},$  
where the factor $\frac{1}{K-k}$ comes the assumption of uniform random useful
piece selection for downloads from the seed.  A young peer that becomes infected
when it has $k$ pieces will eventually release, on average, about $K-k-1$ other peers from the one club.  
Thus, to a first order approximation, for large $n$, the number of peers in the system
behaves like a birth-death process with arrival rate $\lambda$ and state
dependent departure rate $U_s(1 +  \frac{\mu_o}{n\mu}),$   where
$$
\mu_o=\lambda   \sum_{k=0}^{K-2}  \frac{K-k-1}{K-k}.
$$
The elementary theory of birth-death processes shows that a birth-death process with constant birth
rate $\lambda$ and state-dependent death rate $\lambda(1+\frac{c}{n})$ is positive
recurrent if $c > 1$ and null-recurrent if $0<c \leq 1.$  This strongly suggests the
following to be true:
\begin{conjecture}
If $\lambda=U_s$, the process is positive recurrent if $0 < \mu < \mu_o$
and is null recurrent if $\mu > \mu_o.$
\end{conjecture}
We also expect similar results to be true for other piece selection policies, but the value of
$\mu_o$ would depend on the piece selection policy.

\section{Appendix}

\subsection{Stochastic comparison}

A continuous-time random process is said to be {\em c\`{a}dl\`{a}g} if, with the possible exception of a set
of probability zero, the sample paths of the process are right continuous and have finite left limits.

\begin{definition}
Suppose $A=(A_t : t\geq 0)$ and  $B=(B_t : t\geq 0)$  are two random processes, either both
discrete-time random processes, or both continuous time,  c\`{a}dl\`{a}g random processes.
Then $A$ is {\em stochastically dominated} by $B$ if there is a single probability
space $(\Omega, {\cal F}, P)$, and two random processes $\tilde{A}$ and $\tilde{B}$ on
 $(\Omega, {\cal F}, P)$, such that\\
 \begin{description}
 \item (a)  $A$ and $\tilde{A}$ have the same finite dimensional  distributions,   
 \item (b)   $B$ and $\tilde{B}$ have the same finite dimensional  distributions, and
 \item (c)  $P\{\tilde{A}_t \leq \tilde{B}_t~\mbox{for all}~t\}=1.$
 \end{description}
 \end{definition}
Clearly if $A$ is stochastically dominated by $B$, then for any $a$ and $t$, 
  $P\{ A_t \geq a\}  \leq P\{ B_t \geq a\}.$

\subsection{Appendix: Kingman's Moment bound for SII processes}
Let $(X_t :  t\geq 0)$ be a random process with stationary, independent increments with $X_0=0.$
Suppose the sample paths are c\`{a}dl\`{a}g (i.e.  right-continuous with finite left limits).    Suppose $E[X_1^2]$ is
finite, so there are finite constants $\mu$ and $\sigma^2$ such that
 $E[X_t]= \mu t$ and $\mbox{Var}(X_t)= \sigma^2 t$ for all $t\geq 0.$   Let $X^* =\sup_{t\geq 0}  X_t.$

\begin{lemma} (Kingman's moment bound \cite{Kingman62} extended to continuous time)  \label{lemma:Kingman}
Suppose that $\mu < 0.$   Then $E[X^*] \leq   \frac{\sigma^2}{-2\mu}.$
Also, for any $B>0$, $P\{ X^* \geq B \} \leq \frac{\sigma^2}{-2\mu B}.$
\end{lemma}

\begin{proof}
For each integer $n\geq 0$, let $S^n$ denote the random walk process $S^n_k=X_{k2^{-n}}.$
Let $ S^{n*} =\sup_{k \geq 0} S_k.$    By Kingman's moment bound for discrete time
processes,
$$
E[S^{n*}] \leq \frac{\mbox{Var}(S_1^n)}{-2E[S_1^n] }  = \frac{\sigma^2}{-2\mu}
$$
Since  $ S^{n*}$ is nondecreasing in $n$ and converges a.s. to $X^*,$   the first conclusion of
the lemma follows.  The second conclusion follows from the first by Markov's inequality.
\end{proof}

\begin{corollary}  \label{cor:compoundKingman}
Let $C$ be a compound Poisson process with $C_0=0$, with jump times given by a Poisson process
of rate $\alpha$, and jump sizes having mean $m_1$ and mean square value $m_2.$
Then for all $B >0$ and $\epsilon > \alpha m_1$
\begin{equation}   \label{eq:compoundKingman}
P\{ C_t   < B+\epsilon t ~ \mbox{for all}~ t \} \geq 1 -   \frac{\alpha m_2}{2B(\epsilon- \alpha m_1)}   
\end{equation}
\end{corollary}
\begin{proof}
Let $X_t=C_t-\epsilon t.$   Then $X$ satisfies the hypotheses of Lemma \ref{lemma:Kingman}
with $\mu= \alpha m_1-\epsilon$  and $\sigma^2=\alpha m_2.$  
So $P\{ X^* \geq B \} \leq \frac{ \alpha m_2}{-2(\alpha m_1 -\epsilon)B},$ which implies \eqref{eq:compoundKingman}.
\end{proof}

\subsection{A maximal bound for an $\mathbf M/GI/\infty$ queue}

 \begin{lemma}   \label{lemma.mginfty}
 Let $M$ denote the number of customers in an $M/GI/\infty$ queueing system, with arrival rate
 $\lambda$ and mean service time $m.$    Suppose that $M_0=0.$  Then for $B,\epsilon > 0$,
 \begin{equation}  \label{eq.mginfty}
 P\{ M_t \geq B+\epsilon t ~~\mbox{for some}~t\geq 0\}  \leq   \frac{ e^{\lambda(m+1) } 2^{-B} }{1-2^{-\epsilon}}
 \end{equation}
 \end{lemma}
 \begin{proof}
Our idea is to find another $M/GI/\infty$ system whose number of customers sampled at integer times can be used to bound $M.$
Suppose we let every customer for the original process stay in the system for one extra unit time after they have been served.
Let $M^{\sharp}_t$  be the number of customers in this new $M/GI/\infty$ system at time $t$.
Note that $M^{\sharp}$ is also the number in an $M/GI/\infty$ system, with arrival rate $\lambda$ and mean service time $m+1.$
By a well-known property of $M/GI/\infty$ systems,  for any time $t$, $M^{\sharp}_t$ is a Poisson random variable.  Since
the initial state is zero, the mean number in the system at any time $t$ is less than  $\lambda (m+1),$  which is the mean number
in the system in equilibrium.   If $Poi(\mu)$ represents a Poisson random variable with mean $\mu$, then the Chernoff inequality
yields $P\{ Poi(\mu) \geq a \} \leq \exp(\mu(e^\theta -1 ) - \theta a ), $, and taking $\theta=\ln 2$ yields
$P\{ Poi(\mu) \geq a \} \leq e^{\mu} 2^{-a}.$
For any integer $i\geq 1$,  if $t\in (i-1],$  then $M_t \leq M^{\sharp}(i).$   Therefore,
\begin{eqnarray*}
\lefteqn{ P\{ M_t \geq B+\epsilon t ~~\mbox{for some}~t\geq 0\}} \\
&  \leq  & \sum_{i=1}^\infty 
 P\{ M_t \geq B+\epsilon t ~~\mbox{for some}~t \in (i-1,i]  \}  \\
&  \leq  & \sum_{i=1}^\infty 
 P\{ M^\sharp_i \geq B+\epsilon (i-1) \}   \\
&  \leq &   \sum_{i=1}^\infty   e^{\lambda(m+1)}  2^{-(B+\epsilon (i-1)) }\\
& = &  \frac{ e^{\lambda(m+1) } 2^{-B} }{1-2^{-\epsilon}}
\end{eqnarray*}
 
\end{proof}

\subsection{On Busy Periods for M/GI/1 Queues} 

Consider an $M/GI/1$ queue with arrival rate $\lambda.$   
Let $N$ denote the number of
customers served in a busy period, let  $L$ denote the length
of a busy period, and let $X$ denote the service
time of a typical customer.
\begin{lemma}  Let $\rho=\lambda E[X].$  It $\rho < 1$ then
\begin{eqnarray}
E[N]=\frac{1}{1-\rho} ~~~~~~~~ E[N^2]=\frac{1+\lambda^2 \mbox{Var}(X)}{(1-\rho)^3} \label{eq.N} \\
E[L]=\frac{E[X]}{1-\rho}  ~~~  E[L^2] = \frac{E[X^2]}{(1-\rho)^3} ~~~~ \label{eq.L} \\
\mbox{Cov}(N,L) = \frac{\lambda E[X^2]}{(1-\rho)^3}~~~~~~~~~~~~~~~
\end{eqnarray}
\end{lemma}
The lemma can be proved by the well-known branching process
method. Let $X$ denote the service time of a customer starting a new
busy period.  Let $Y$ denote the number of arrivals while the first customer is
being served.  Then, given $X=x$, the conditional distribution
of $Y$ is Poisson with mean $\lambda x.$   View any customer
in the busy period that arrives after the first customer, to be the
offspring of the customer in the server at the time of  arrival.
This gives the well known representation for $N$ and $L$:
\begin{eqnarray*}
N & = & 1 + \sum_{i=1}^Y   N_i  \\
L & = & X+ \sum_{i=1}^Y L_i
\end{eqnarray*}
where $(N_i,L_i), i\geq 1$ is a sequence of independent random
2-vectors such that for each $i$, $(N_i,L_i)$ has the same
distribution as $(N,L).$   Using Wald's identity, these
equations can be used to prove the lemma.

\subsection{Foster-Lyapunov stability criterion}

\begin{prop}   \label{cor.FosterCompContinuous}
{\em Combined Foster-Lyapunov stability criterion and moment bound--continuous time} (See  \cite{Hajek567,MeynTweedie09}.)
Suppose $X$ is a continuous-time, irreducible Markov process on a countable state space ${\cal S}$ with generator matrix $Q.$
Suppose $V$, $f$, and $g$ are nonnegative functions on $\cal S$ such that
$QV(\mathbf{x})   \leq -f(\mathbf{x}) +g(\mathbf{x})$ for all $\mathbf{x}\in {\cal S}$, and, for some $\delta > 0$,  the set $C$ defined by
$C=\{ \mathbf{x} : f(\mathbf{x}) < g(\mathbf{x})+\delta\}$ is finite.  Suppose also that $\{ \mathbf{x} : V(\mathbf{x}) \leq K\}$ is finite for all $K$.
Then $X$ is positive recurrent and, if $\pi$ denotes the equilibrium distribution,   $\sum_\mathbf{x}  f(\mathbf{x})\pi(\mathbf{x})
 \leq \sum_\mathbf{x}  g(\mathbf{x})\pi(\mathbf{x})$.
\end{prop}


\begin{thebibliography}{}
\ifx \url   \undefined \def \url#1{#1}   \fi

\bibitem{AhlswedeCaiLiYeung}
\textsc{Ahlswede, R.}, \textsc{Cai, N.}, \textsc{Li, S.-Y.}, \textsc{and}
  \textsc{Yeung, R.} (2000).
\newblock Network information flow.
\newblock \emph{IEEE Transactions on Information Theory\/}~\textbf{46},~4
  ({jul}), 1204 --1216.

\bibitem{Cohen03}
\textsc{Cohen, B.} (2003).
\newblock Incentives build robustness in {BitTorrent}.
\newblock \emph{P2PECON Workshop\/}.

\bibitem{DebMedardChoute06}
\textsc{Deb, S.}, \textsc{M\'{e}edard, M.}, \textsc{and} \textsc{Choute, C.}
  (2006).
\newblock Algebraic gossip: A network coding approach to optimal multiple rumor
  mongering.
\newblock \emph{IEEE Transactions on Information Theory\/}~\textbf{52},~6
  (June), 2486--2502.

\bibitem{GkantsidisRodriguez05}
\textsc{Gkantsidis, C.} \textsc{and} \textsc{Rodriguez, P.} (2005).
\newblock Network coding for large scale content distribution.
\newblock In \emph{Proceedings INFOCOM 2005}.  Vol.~\textbf{4}. 2235 -- 2245
  vol. 4.

\bibitem{Hajek567}
\textsc{Hajek, B.}
\newblock Notes for {ECE} 567: Communication network analysis.
\newblock Available at www.illinois.edu/~b-hajek.

\bibitem{Kingman62}
\textsc{Kingman, J.} (1962).
\newblock Some inequalities for the queue {GI}/{G}/1.
\newblock \emph{Biometrika\/}~\textbf{49},~3/4, 315--324.

\bibitem{Kurtz81}
\textsc{Kurtz, T.~G.} (1981).
\newblock \emph{Approximation of population processes}. CBMS-NSF Regional
  Conference Series in Applied Mathematics, Vol.~\textbf{36}.
\newblock Society for Industrial and Applied Mathematics (SIAM), Philadelphia,
  Pa.


\bibitem{LeskelaRobertSimatos10}
\textsc{Leskel\"{a}, L.}, \textsc{Robert, P.}, \textsc{and} \textsc{Simatos,
  F.} (2010).
\newblock Interacting branching processes and linear file-sharing networks.
\newblock \emph{Advances in Applied Probability\/}~\textbf{42},~3, 834--854.

\bibitem{MassoulieVojnovic05}
\textsc{Massouli\'{e}, L.} \textsc{and} \textsc{Vojnovi\'{c}, M.} (2005).
\newblock Coupon replication systems.
\newblock In \emph{SIGMETRICS '05: Proceedings of the 2005 ACM SIGMETRICS
  international conference on Measurement and modeling of computer systems}.
  ACM, New York, NY, USA, 2--13.


\bibitem{MassoulieVojnovic08}
\textsc{Massouli\'{e}, L.} \textsc{and} \textsc{Vojnovi\'{c}, M.} (2008).
\newblock Coupon replication systems.
\newblock \emph{IEEE/ACM Trans. Networking\/}~\textbf{16},~3, 603--616.

\bibitem{Menasche_etal10}
\textsc{Menasch{\'e}, D.~S.}, \textsc{de~Arag{\~a}o~Rocha, A.~A.},
  \textsc{de~Souza~e Silva, E.}, \textsc{Le{\~a}o, R. M.~M.}, \textsc{Towsley,
  D.~F.}, \textsc{and} \textsc{Venkataramani, A.} (2010).
\newblock Estimating self-sustainability in peer-to-peer swarming systems.
\newblock http://arxiv.org/abs/1004.0395.

\bibitem{MeynTweedie09}
\textsc{Meyn, S.} \textsc{and} \textsc{Tweedie, R.} (2009).
\newblock \emph{Markov Chains and Stochastic Stability (Cambridge Mathematical
  Library)}, 2 ed.
\newblock Cambridge University Press.

\bibitem{NorrosReittuEirola09}
\textsc{Norros, I.}, \textsc{Reittu, H.}, \textsc{and} \textsc{Eirola, T.}
  (2009).
\newblock On the stability of two-chunk file-sharing systems.
\newblock Available on arXiv:0910.5577.

\bibitem{QiuSrikant04}
\textsc{Qiu, D.} \textsc{and} \textsc{Srikant, R.} (2004).
\newblock Modeling and performance analysis of {B}it{T}orrent-like peer-to-peer
  networks.
\newblock In \emph{SIGCOMM '04}. ACM, New York, NY, USA, 367--378.


\bibitem{YangDeVeciana04}
\textsc{Yang, X.} \textsc{and} \textsc{de~Veciana, G.} (2004).
\newblock Service capacity of peer to peer networks.
\newblock In \emph{IEEE INFOCOM}. 1--11.

\bibitem{YangDeVeciana06}
\textsc{Yang, X.} \textsc{and} \textsc{de~Veciana, G.} (2006).
\newblock Performance of peer-to-peer networks: Service capacity and role of
  resource sharing policies.
\newblock \emph{Performance Evaluation\/}~\textbf{63},~3, 175--194.

\bibitem{ZhuHajek11}
\textsc{Zhu, J.} \textsc{and} \textsc{Hajek, B.} (2011).
\newblock Stability of a peer-to-peer communication system.
\newblock In \emph{Proceedings SIGACT-SIGOPS Symposium on Principles of
  Distributed Computing}.

\end{thebibliography}

\end{document}